\pdfoutput = 1

\def\llncs{0}
\def\mnotes{0}
\ifnum\llncs=1
    \documentclass[envcountsect,runningheads]{llncs}
    \usepackage{pifont}
\else
    \documentclass[11pt]{article}
    \usepackage{fullpage,times, amsthm,latexsym,amssymb,graphicx}
\def\colorson{0}

\usepackage{color,amsfonts,amsmath,url}
\ifnum\llncs=0

\usepackage{wrapfig}

\usepackage[colorlinks=true,breaklinks=true,linkcolor=black,
citecolor=black,filecolor=black,menucolor=black,
pagecolor=black,urlcolor=black]{hyperref}

\advance \topmargin by -\headheight

\advance \topmargin by -\headsep

\fi %matches ifnum\llncs=0

\ifnum\colorson=1
\newenvironment{todo}{\noindent
\sf \footnotesize \textcolor{blue}{To go here}:
\begin{CompactItemize}\color{blue}}
{\color{black}\end{CompactItemize}\rm \normalsize}
\else

\fi
\ifnum\llncs=0

\newtheorem{theorem}{Theorem}[section]
\newtheorem{lemma}[theorem]{Lemma}
\newtheorem{sublemma}[theorem]{SubLemma}
\newtheorem{prop}[theorem]{Proposition}

\newtheorem{claim}[theorem]{Claim}
\newtheorem{myclaim}[theorem]{Claim}
\newtheorem{fact}[theorem]{Fact}
\newtheorem{corollary}[theorem]{Corollary}
\newtheorem{conjecture}{Conjecture}
\newtheorem{observation}{Observation}

\newtheorem{definition}{Definition}[section]
\newtheorem{defn}[definition]{Definition}

\newtheorem{protocol}{Protocol}%[section]

\theoremstyle{definition}

\theoremstyle{remark}

\newtheorem{myremark}{Remark} [section]
\newenvironment{remark}{\begin{myremark}}{$\diamondsuit$\end{myremark}}

\newtheorem{myexample}{Example}

\else
\newtheorem{claim}[section]{Claim}

\newtheorem{remark}[section]{Remark}

\newtheorem{definition}{Definition}[section]

\fi
\fi

\usepackage{amsfonts,amsmath, amsopn, amssymb, epsfig, graphicx, multirow, verbatim}
\usepackage{color,bm}
\ifnum\mnotes=0
\newcommand{\mnote}[1]{}
\else
\newcounter{mynotes}
\newcommand{\mnote}[1]{\addtocounter{mynotes}{1}{{\bf !}}%
\marginpar{\scriptsize  {\arabic{mynotes}.\ {\sf \textcolor{red}{#1}}}}}
\fi

\newcommand{\snote}[1]{\mnote{S: #1}}

\newcommand{\gnote}[1]{\mnote{G: #1}}

\newcommand{\secref}[1]{Section~\ref{sec:#1}}

\newcommand{\ith}{i^{th}}

\newcommand{\kth}{k^{th}}

\newenvironment{myproof}{\begin{proof}%\small
}{\ifnum\llncs=1
{~}\qed
\fi
\end{proof}}
\newenvironment{proofof}[1]{\begin{myproof}[Proof of {#1}]%\small
}{\end{myproof}}

\newcommand{\thmref}[1]{Theorem~\ref{thm:#1}}

\newcommand{\lemref}[1]{Lemma~\ref{lem:#1}}

\newcommand{\corref}[1]{Corollary~\ref{cor:#1}}
\newcommand{\claimref}[1]{Claim~\ref{claim:#1}}

\newcommand{\ignore}[1]{}

\DeclareMathOperator{\polylog}{polylog}
\DeclareMathOperator{\loglog}{\log\log}
\DeclareMathOperator{\poly}{poly}

\newcommand{\etal}{{\it et~al.}}

\newcommand{\calh}{{\cal H}}
\newcommand{\calg}{{\cal G}}

\newcommand{\fone}{{\vec{1}}}
\newcommand{\fzero}{{\vec{0}}}

\newcommand{\Vol}{V}
\newcommand{\dprec}{\prec\!\prec}

\begin{document}

\ifnum\llncs=1
    \bibliographystyle{splncs}
\else
    \bibliographystyle{abbrv}
\fi

\title{Steiner Transitive-Closure Spanners of $d$-Dimensional Posets}

\ifnum\llncs=1

\author{
Piotr Berman\inst{1}
and Arnab Bhattacharyya\inst{2}
and Elena Grigorescu\inst{3}
and Sofya Raskhodnikova\inst{1}\thanks{Supported by NSF CAREER award 0845701.}
and David P. Woodruff\inst{4}
and Grigory Yaroslavtsev\inst{1}
}

\institute{
Pennsylvania State University, USA. {\tt \{berman, sofya, grigory\}@cse.psu.edu}.
S.R. and G.Y. are supported by NSF / CCF CAREER award 0845701. G.Y. is also supported by University Graduate Fellowship and College of Engineering Fellowship.
\and Massachusetts Institute of Technology, USA. {\tt abhatt@mit.edu}.
\and Georgia Institute of Technology, USA. {\tt elena@cc.gatech.edu}.
\and IBM Almaden Research Center, USA. { \tt dpwoodru@us.ibm.com}.
}
\else
   \author{
        Piotr Berman\thanks{Pennsylvania State University, USA. {\tt \{berman, sofya, grigory\}@cse.psu.edu}.
        S.R. and G.Y. are supported by NSF / CCF CAREER award 0845701. G.Y. is also supported by University Graduate Fellowship and College of Engineering Fellowship.}
        \and Arnab Bhattacharyya\thanks{Massachusetts Institute of Technology, USA. {\tt abhatt@mit.edu }}
        \and Elena Grigorescu\thanks{Georgia Institute of Technology, USA. {\tt elena@cc.gatech.edu}.
        Supported in part by NSF award CCR-0829672 and NSF award 1019343 to the Computing Research Association for the Computing Innovation Fellowship Program.}
        \and Sofya Raskhodnikova\protect\footnotemark[1]
        \and David P. Woodruff\thanks{IBM Almaden Research Center, USA. { \tt dpwoodru@us.ibm.com}.}
        \and Grigory Yaroslavtsev\protect\footnotemark[1]
   }
\fi

\maketitle

\begin{abstract}
Given a directed graph $G = (V,E)$ and an integer $k \geq 1$, a {\em $k$-transitive-closure-spanner} ($k$-TC-spanner) of $G$ is a directed graph $H = (V, E_H)$ that has (1)~the same transitive-closure as $G$ and (2) diameter at most~$k$.  In some applications, the shortcut paths added to the graph in order to obtain small diameter can use Steiner vertices, that is, vertices not in the original graph $G$. The resulting spanner is called a {\em Steiner transitive-closure spanner} (Steiner TC-spanner).

Motivated by applications to property reconstruction and access control hierarchies, we concentrate on Steiner TC-spanners of directed acyclic graphs or, equivalently, partially ordered sets. In these applications, the goal is to find a sparsest Steiner $k$-TC-spanner of a poset $G$ for a given $k$ and $G$.
The focus of this paper is the relationship between the dimension of a poset and the size of its sparsest Steiner TC-spanner. The dimension of a poset $G$ is the smallest $d$ such that $G$ can be embedded into a $d$-dimensional directed hypergrid via an order-preserving embedding.

We present a nearly tight lower bound on the size of Steiner 2-TC-spanners of $d$-dimensional directed hypergrids. It implies better lower bounds on the complexity of local reconstructors of monotone functions and functions with low Lipschitz constant. The proof of the lower bound constructs a dual solution to a linear
programming relaxation of the Steiner 2-TC-spanner problem.
We also show that one can efficiently construct a Steiner $2$-TC-spanner, of size matching the lower bound, for any low-dimensional poset.
\ignore{We also give the first Steiner 2-TC-spanner construction for
$d$-dimensional posets and show that our construction gives
2-TC-spanners of nearly-optimal size.} Finally, we present a lower bound on
the size of Steiner $k$-TC-spanners of $d$-dimensional posets that shows that
the best-known construction, due to De Santis~\etal, cannot be improved significantly.
\end{abstract}

\thispagestyle{empty}
\setcounter{page}{0}
\newpage
\section{Introduction}\label{sec:intro}
Graph spanners were introduced in the context of distributed
computing~\cite{PelegSchaffer}, and since then have found numerous
applications.
Our focus is on transitive-closure spanners, introduced explicitly in \cite{tc-spanners-soda}, but studied prior to that in many different contexts~\cite{CFL83a, CFL83b, Yao82,AlonSchieber87,cha87,tho92,BTS94,thorup95,thorup97,DGLRRS99,h03,AFB05,ABF06,atallah-journal}.

%Given a directed graph $G = (V,E)$ and an integer $k \geq 1$, a $k$-transitive-closure-spanner ($k$-TC-spanner) of $G$ is a directed graph $H = (V, E_H)$ that has (1)~the same transitive-closure as $G$ and (2) diameter at most~$k$.
Given a directed graph $G=(V,E)$ and an integer $k \geq 1$, a
\textbf{$k$-transitive-closure-spanner} ($k$-TC-spanner) of $G$ is a directed graph $H=(V,E_H)$
satisfying:
(1) $E_H$ is a subset of the edges in the transitive closure of $G$;
(2) for all vertices $u,v\in V$, if $d_G(u,v) < \infty$ then $d_H(u,v) \leq k$.
That is, a $k$-TC-spanner is a graph with a small diameter
that preserves the connectivity of the original graph.
%%
%It can also be viewed as the $k$-multiplicative spanner of the transitive closure of $G$.
The edges from the transitive closure of $G$ that are added to $G$ to obtain a TC-spanner are called {\em shortcut edges} and the parameter $k$ is called the {\em stretch}.

TC-spanners have numerous applications, and there has been lots of work on finding sparse TC-spanners for specific graph families. (See~\cite{Ras-survey10} for a survey.) In some applications of TC-spanners (in particular, to access control hierarchies \cite{ABF06,AFB05,SFM07,atallah-journal}), the shortcuts can use {\em Steiner} vertices, that is, vertices not in the original graph $G$. The resulting spanner is called a {\em Steiner TC-spanner}.
\begin{definition}[Steiner TC-spanner]\label{def:steiner-tc-spanner}
Given a directed graph $G=(V,E)$ and an integer $k \geq 1$, a
\textbf{Steiner $k$-transitive-closure-spanner} (\textbf{Steiner $k$-TC-spanner}) of $G$ is a directed graph $H=(V_H,E_H)$
satisfying:
%\begin{enumerate}
%\item
{\em(1)} $V\subseteq V_H$;
%\item $E_H$ is a subset of the edges in the transitive closure of $G$.
%\item
{\em(2)} for all vertices $u,v\in V$,
%\begin{itemize}
%\item[]
%\begin{center}
if $d_G(u,v) < \infty$ then $d_H(u,v) \leq k$ and
%\item[]
if $d_G(u,v) = \infty$ then $d_H(u,v) = \infty$.
%\end{center}
%\end{itemize}
%\end{enumerate}
Vertices in $V_H\backslash V$ are called {\em Steiner vertices}.%, and the parameter $k$ is called the \textbf{stretch}.
\end{definition}
For some graphs, Steiner TC-spanners can be significantly sparser than ordinary TC-spanners. For example, consider a  complete bipartite graph $K_{\frac n2, \frac n 2}$ with $n/2$ vertices in each part and all edges directed from the first part to the second.  Every ordinary 2-TC-spanner of this graph has $\Omega(n^2)$ edges. However,  $K_{\frac n2, \frac n 2}$ has a Steiner 2-TC-spanner with $n$ edges: it is enough to add one Steiner vertex $v$, edges to $v$ from all nodes in the left part, and edges from $v$ to all nodes in the right part. Thus, for $K_{\frac n2, \frac n 2}$ there is a linear gap between the size of the sparsest Steiner 2-TC-spanner and the size of an ordinary 2-TC-spanner.
%\snote{Generalize to show the gap for general k: multi-layered complete graph}
%\begin{figure*}
\begin{center}
\ifnum\llncs=1
\includegraphics[width=\columnwidth]{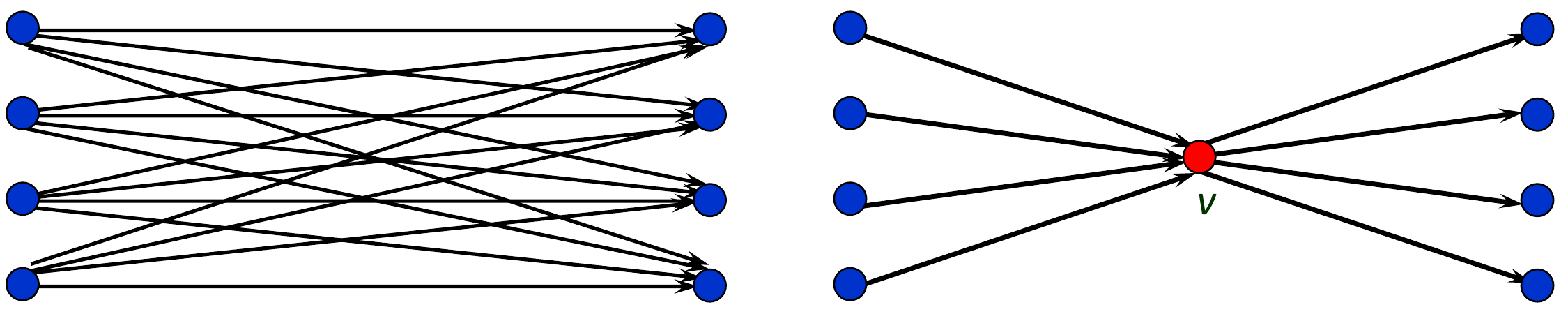}
\else
\includegraphics[width=6in]{tc-spanner-survey-illustration_steiner-tc-spanner}
\fi
%\caption{Directed line $L_n$ and its transitive closure TC($L_n$)}
%\label{fig:tc-line}
\end{center}
%\end{figure*}

We concentrate on Steiner TC-spanners of directed {\em acyclic} graphs (DAGs) or, equivalently, partially ordered sets (posets) because they represent the most interesting case in applications of TC-spanners.
In addition, there is a reduction from constructing TC-spanners of graphs
 with cycles to constructing TC-spanners of DAGs, with a small loss in stretch (\cite{Ras-survey10}, Lemma~3.2), which also applies to Steiner TC-spanners.

The goal of this work is to understand the minimum number of edges
needed to form a Steiner $k$-TC-spanner of a given graph $G$ as a
function of $n$, the number of nodes in $G$. More specifically, motivated by applications to access control hierarchies \cite{ABF06,AFB05,SFM07,atallah-journal} and property reconstruction \cite{bgjjrw2010,lipschitz}, described in Section~\ref{sec:applications}, we study the relationship between the dimension of a poset and the size of its sparsest Steiner TC-spanner. The {\em dimension} of a poset $G$ is the smallest $d$ such that $G$ can be embedded into a $d$-dimensional directed hypergrid via an order-preserving embedding. (See Definition~\ref{def:poset-dimension}).  Atallah~\etal~\cite{atallah-journal}, followed by De Santis~\etal~\cite{SFM07},
use Steiner TC-spanners in key management schemes for access control hierarchies. They argue that many access control hierarchies are low-dimensional posets that come equipped with an embedding
demonstrating low dimensionality. For this reason, we focus on the setting where the dimension $d$ is small relative to the
number of nodes $n$.

We also study the size of sparsest (Steiner) 2-TC-spanners of specific posets of dimension $d$, namely, $d$-dimensional directed hypergrids. Our lower bound on this quantity improves the result in~\cite{bgjjrw2010} and nearly matches the upper bound from that paper. It implies that our construction of Steiner 2-TC-spanners of $d$-dimensional posets cannot be improved significantly. It also has direct implications for property reconstruction. The focus on stretch $k=2$ is motivated by both applications.

\ifnum\llncs=1
Several classes of posets, in addition to low-dimensional hypergrids, are known to have small dimension. For example, if a {\em planar poset}, that is, a poset with a planar Hasse diagram,  has
{\em both} a minimum {\em and} a maximum element, it has dimension at most $2$.\snote{Old: cite handbook of combinatorics 1 vol 1; all the references below are mentioned there} A planar poset with {\em either} a minimum or a maximum element has dimension at most $3$ \cite{TrotterMoore}. (In general, however, a planar poset can have an arbitrary dimension \cite{kelley81}).
%Do we need that?
%Viewing ${\cal H}_{m,d}$ as a partial ordering of set inclusion, Dushnik and Miller \cite{dush-miller} show that ${\cal H}_{m,d}$ has dimension at least $d$.  ${\cal H}_{m,d}$ has dimension at most $d$ by definition.
One can also bound the dimension in terms of cardinality of the poset: every poset of cardinality $n\geq 4$ has dimension at most $n/2$
\cite{Hiraguchi51}. Also, if every element in an $n$-element poset has at most $u$ points above it,
then its dimension is at most $2(u+1)\log n +1$ \cite{furedikahn}.  Thus, posets with low dimension
occur quite naturally in a variety of settings.
%
% Atallah~\etal give a nice summary of relevant work on poset dimension; we mention the most salient results, all due to Schnyder~\cite{Schnyder89}. A poset $G$ with planar transitive reduction has dimension at most 3. Moreover, the corresponding embedding of $G$ into a hypergrid is computable in linear time. If the transitive reduction of $G$ is 4-colorable, then $G$ has dimension at most 4.
%For instance, planar posets with a maximal or minimal element have dimension at most $3$ by a
%result of Trotter and Moore  and they admit sparse TC-spanner constructions as indicated
%above.
\fi

%In addition to studying general posets of dimension $d$, we pay special attention to 2-TC-spanners of the $d$-dimensional hypercube and hypergrid. This is motivated by applications to property testing.

\subsection{Our Results}\label{sec:results}
\paragraph{Steiner 2-TC-spanners of directed $d$-dimensional grids.} The {\em directed hypergrid}, denoted ${\cal H}_{m,d}$, has vertex set\footnote{For a positive integer $m$, we denote $\{1,\dots, m\}$ by $[m]$.}
$[m]^d$ and edge set $\{(x,y) : \exists \text{ unique } i \in [d]\text{
  such that } y_i-x_i = 1 \text{ and if } j\neq i, y_j = x_j\}$. We observe (in \corref{steiner-points-do-not-help}) that for the grid ${\cal H}_{m,d}$, Steiner vertices do not help to create sparser $k$-TC-spanners. In \cite{bgjjrw2010}, it was shown that for $m\geq 3$, sparsest (ordinary) 2-TC-spanners of $\calh_{m,d}$ have size at most $m^d \log^d m$ and at least $\Omega\left(\frac{m^d \log^d m}{(2d\log \log m)^{d-1}}\right)$. They also give tight upper and lower bounds for the case of constant $m$ and large $d$. Our first result is an improvement on the lower bound for the hypergrid for the case when $m$ is significantly larger than $d$.

\begin{theorem}\label{thm:grid-lb}
%For $m > e^{2d}$,
Every (Steiner) 2-TC-spanner of $\mathcal{H}_{m,d}$ has
$\displaystyle\Omega\Big(\frac{m^d (\ln m-1)^d}{(4\pi)^d}\Big)$ edges.
\end{theorem}

The proof of \thmref{grid-lb} constructs a dual solution to a linear
programming relaxation of the Steiner 2-TC-spanner problem. We consider an integer linear program for the sparsest 2-TC-spanner of $\calh_{m,d}$. Our program is a special case of a more general linear program for the sparsest directed $k$-spanner of an arbitrary graph $G$, used in \cite{tc-spanners-soda} to obtain an approximation algorithm for that problem. As explained in~\cite{tc-spanners-soda}, the general program has an integrality gap of $\Omega(n)$. However, we show that for our special case the integrality gap is small and, in particular, does not depend on $n$.
%%Sofya: explanation for the proof without explicitly writing out the dual (in case we want to revert).
%Specifically, we find a solution to the dual linear program by summing up the constraints of the original program with carefully selected coefficients. Our coefficients have a combinatorial interpretation: they are expressed in terms of the {\em volume} of $d$-dimensional {\em boxes} contained in $\calh_{m,d}$.  For example, the constraint that enforces the existence of a length-2 path from $u$ to $v$ in the 2-TC-spanner is initially assigned a coefficient which is inversely proportional to the number of nodes on the paths from $u$ to $v$. The final sum of the constraints is bounded by an integral which, in turn, is bounded by an expression depending only on the dimension $d$.
Specifically, we find a solution to the dual linear program by selecting initial values that have a combinatorial interpretation: they are expressed in terms of the {\em volume} of $d$-dimensional {\em boxes} contained in $\calh_{m,d}$.  For example, the dual variable corresponding to the constraint that enforces the existence of a length-2 path from $u$ to $v$ in the 2-TC-spanner is initially assigned a value inversely proportional to the number of nodes on the paths from $u$ to $v$. The final sum of the constraints is bounded by an integral which, in turn, is bounded by an expression depending only on the dimension $d$.

We note that the best lower bound known previously~\cite{bgjjrw2010} was proved by a long and sophisticated combinatorial argument that carefully balanced the number of edges that stay within different parts of the hypergrid and the number of edges that cross from one part to another. Our linear programming argument can be thought of as assigning types to edges based on the volume of the boxes they define, and automatically balancing the number of edges of different types by selecting the correct coefficients for the constraints corresponding to those edges.

\paragraph{Steiner TC-spanners of general $d$-dimensional posets.}
We continue the study of the number of edges in a sparsest Steiner $k$-TC-spanner of a poset as a function of its dimension, following %Atallah \etal~
\cite{atallah-journal} and %De Santis \etal~
\cite{SFM07}.
Observe that the only poset of dimension $1$ is the directed line $\calh_{n,1}$.
TC-spanners of the directed lines %and directed trees
were discovered under many different guises. They were studied implicitly in \cite{AlonSchieber87,AFB05,cha87,DGLRRS99,Yao82} and
explicitly in \cite{BTS94,thorup97}.
Alon and Schieber \cite{AlonSchieber87} implicitly showed that, for constant $k$,
the size of the sparsest $k$-TC-spanner of the directed line  is $\Theta(n\cdot \lambda_k(n)),$ where $\lambda_k(n)$ is the $\kth$-row
inverse Ackermann function.\footnote{The  {\em Ackermann function} \cite{a28} is defined by: $A(1, j)=2^j$, $A(i+1, 0)=A(i, 1), $ $A(i+1, j+1)=A(i, 2^{2^{A(i+1, j)}})$.
The inverse Ackermann function is $\alpha(n)=\min \{i: A(i, 1)\geq n\}$ and the $\ith$-row inverse is $\lambda_i(n)=\min \{j: A(i, j)\geq n\}$. Specifically, $\lambda_2(n)=\Theta(\log n)$, $\lambda_3(n)=\Theta(\loglog n)$ and $\lambda_4(n)=\Theta(\log^* n)$.}

Table~\ref{table:results} compares old and new results for $d\geq
2$. $S_k(G)$ denotes the number of edges in the sparsest Steiner $k$-TC-spanner of $G$. The upper bounds hold for all posets of dimension~$d$. The lower bounds mean that there is a poset of dimension $d$ for which every Steiner $k$-TC-spanner has the specified number of edges.

\begin{table}
\begin{tabular}{|l|l|}
  \hline
  %$k$ & Bounds on $S_k(G)$ in \cite{atallah-journal} \\
  Stretch $k$ & Prior bounds on $S_k(G)$ \\
  \hline
  \hline
  $2d-1$ & $O(n^2)$\hfill  \cite{atallah-journal} \\
   $2d-2+t$ for $t\geq 2$ & $O(n (\log^{d-1} n) \lambda_{t} (n))$ \hfill  \cite{atallah-journal}\\
   $2d+O(\log^* n)$ & $O(n \log^{d-1} n)$ \hfill  \cite{atallah-journal}\\
   \hline
 \multirow{2}{*}{3} & $O(n\log^{d-1} n \log \log n)$ \\
   & \hfill for fixed $d$ \ \cite{SFM07}\\
  \hline
\end{tabular}
\hspace*{1mm}
\begin{tabular}{|l|l|l|}
  \hline
  Stretch $k$ & \multicolumn{2}{|c|}{Our bounds on $S_k(G)$}\\
  \hline
  \hline
  \multirow{2}{*}{2} & $O(n \log^{d} n)$ & $\Omega \left (n \left (\frac{\log n}{cd} \right )^d \right)$ \\
  & \hfill for all $d$& \hfill for a fixed $c > 0$ and all  $d$ \\
  \hline
   \multirow{2}{*}{$\geq 3$} &   & $\Omega(n \log^{\lceil (d-1)/k \rceil} n)$ \\
   &   & \hfill for fixed $d$\\
  \hline
\end{tabular}
\vspace*{3mm}
\caption{\ifnum\llncs=0
The size of the sparsest Steiner $k$-TC-spanner for $d$-dimensional posets on $n$ vertices for $d\geq 2$
\else
Steiner $k$-TC-spanner sizes for $d$-dimensional posets on $n$ vertices for $d\geq 2$
\fi
}\label{table:results}\end{table}

Atallah \etal\ construct Steiner $k$-TC-spanners with $k$ proportional to
$d$.  De Santis \etal\ improved their construction for constant
$d$. They achieve $O(3^{d-t}nt\log^{d-1} n\log\log n)$ edges for odd stretch $k=2t+1$, where $t \in [d]$. In particular, setting $t=1$ gives $k = 3$ and $O(n \log^{d-1} n \log \log n)$
edges.
% This number of edges is an $O(\log \log n)$ factor worse than in the best construction of Atallah \etal, but the stretch is improved  from $\Theta(d)$ to $3$.

\iffalse
In contrast, we achieve $k=3$ with only $O(n (\log \log n)^d \log n)$ edges, a significant improvement over previous work for all $k \geq 3$.
\fi

 We present the first construction of Steiner 2-TC-spanners for
 $d$-dimensional posets. In our construction, the spanners have $O(n
 \log^d n)$ edges, and the length-$2$  paths  can be found in $O(d)$
 time. This result is stated in Theorem~\ref{thm:2tc-ub} (Sect.~\ref{sec:prelim}).
%Our results drastically improve the performance of the key management schemes for access hierarchies developed by Atallah \etal\ and De Santis \etal\ %no period
% We prove a lower bound that matches this construction to
%within $(\log\log n)^d$ factors when $d$ is constant.
Our construction takes as part of the input an explicit
 embedding of the poset
into a $d$-dimensional grid. (Finding such an embedding is NP-hard
\cite{yannakakis82}.) %in the worst-case.%
%Atallah claim without ref "even approximating it to within a constant factor is not known to be in P"
%Our Steiner $2$-TC-spanner construction works for all $d$.

Note that the Steiner vertices used in our construction for $d$-dimensional posets are necessary to obtain sparse TC-spanners. Recall our example of a bipartite graph $K_{\frac n2, \frac n 2}$ for which every 2-TC-spanner required $\Omega(n^2)$ edges. $K_{\frac n2, \frac n 2}$ is a poset of dimension $2$,  and thus, by the upper bound in Theorem~\ref{thm:2tc-ub}, has a Steiner $2$-TC-spanner of size $O(n\log^2 n)$. (As we mentioned before, for this graph there is an even better  Steiner $2$-TC-spanner with $O(n)$ edges.) To see that $K_{\frac n2, \frac n 2}$ is embeddable into a
$[n]\times [n]$ grid, map each of the $n/2$ left vertices of $K_{\frac n2, \frac n 2}$ to a distinct grid vertex in the set of incomparable vertices $\{ (i, n/2+1-i): i\in [n/2]\}$, and similarly map each right vertex to a distinct vertex  in the set $\{(n+1-i, i+n/2): i\in [n/2] \}$. It is easy to see that this
is a proper embedding.

%
%Moreover, our construction for $k=3$ uses fewer edges than all constructions presented by Attalah~\etal or those presented by Santis~\etal~\cite{Santisetal}.
Theorem~\ref{thm:grid-lb} implies
%While previous work did not provide any lower bounds, we demonstrate
that there is an absolute constant $c > 0$ for which
our upper bound for  $k = 2$  is  tight within an $O((cd)^d)$ factor%for every $d \leq c' \log n$
, showing that no drastic improvement in the upper
bound is possible.
To obtain a bound in terms of the number $n$ of vertices and dimension $d$, substitute $m^d$ with $n$ and $\ln m$ with $(\ln n)/d$ in the theorem statement. This gives the following corollary.
\begin{corollary}\label{cor:lb-steiner-2tc}
There is an absolute constant $c > 0$ for which for all
$d \geq 2$,
%For all constant $d\geq 2$,
there exists a $d$-dimensional poset $G$ on $n$ vertices
such that every Steiner $2$-TC-spanner
of $G$ has $\Omega\Big(n\big(\frac{\log n}{cd}\big)^d\Big)$ edges.
\end{corollary}

In addition, we prove a lower bound for all
constant $k>2$ and constant dimension $d$, which qualitatively matches
known upper bounds. It shows that, in particular, every Steiner 3-TC-spanner has size $\Omega(n\log n)$, and even with
significantly larger constant stretch,
every Steiner TC-spanner has size $n \log^{\Omega(d)} n$.
\begin{theorem}\label{thm:lb-steiner-2dim-k>2}
For all constant $d\geq 2$, there exists a $d$-dimensional poset $G$ on $n$ vertices
such that for all $k \geq 3$, every Steiner $k$-TC-spanner
of $G$ has $\Omega(n \log^{\lceil (d-1)/k \rceil} n)$ edges.
\end{theorem}
\noindent This theorem (proved in Section~\ref{sec:lb-posetk>2}) greatly
improves upon the previous $\Omega(n \log \log n)$ bound, which follows
trivially from known lower bounds for a $3$-TC-Spanner of a directed line.

The lower bound on the size of a Steiner $k$-TC-spanner for $k\geq 3$ is proved by the probabilistic method. We observe that using the hypergrid as an example of a poset with large Steiner $k$-TC-spanners for $k>2$ would yield a much weaker lower bound because it is known that $\calh_{m,d}$ has a 3-TC-spanner of size $O((m\loglog m)^d)$ and, more generally, a $k$-TC-spanner of size $O((m\cdot \lambda_k(m))^d),$ where $\lambda_k(m)$ is the $\kth$-row inverse Ackermann function \cite{bgjjrw2010}.
Instead, we construct an $n$-element poset embedded in $\calh_{n,d}$ using the following randomized procedure: all poset elements differ on coordinates in dimension 1, and for each element, the remaining $d-1$ coordinates are chosen uniformly at random from $[n]$. We consider a set of partitions of the underlying hypergrid into $d$-dimensional boxes, and carefully count the expected number of edges in a Steiner $k$-TC-spanner that cross box boundaries for each partition. Then we show that each edge was counted only a small number of times, proving that the expected number of edges in a Steiner $k$-TC-spanner is large. We conclude that some poset attains the expected number of edges.
%proceeds by carefully counting edges in stages.

\paragraph{Organization.} We explain applications of Steiner TC-spanners in Section~\ref{sec:applications}. Section~\ref{sec:prelim} gives basic definitions and observations. In particular, our construction of sparse Steiner 2-TC-spanners for $d$-dimensional posets (the proof of Theorem~\ref{thm:2tc-ub}) is presented there.
Our lower bounds are the technically hardest part of this paper. The lower bound for the hypergrid for $k=2$ (Theorem~\ref{thm:grid-lb}) is proved in Section~\ref{sec:grid-lower-bounds}. The lower bound for $k>2$ (Theorem~\ref{thm:lb-steiner-2dim-k>2}) is presented in Section~\ref{sec:lb-posetk>2}.

\subsection{Applications}\label{sec:applications}
Numerous applications of TC-spanners are surveyed in~\cite{Ras-survey10}. We focus on two of them: property reconstruction, described in \cite{bgjjrw2010,lipschitz}, and key management for access control hierarchies, described in \cite{ABF06,AFB05,SFM07,atallah-journal,tc-spanners-soda}.

\paragraph{Property Reconstruction.}
%The model of property-preserving data reconstruction was introduced in~\cite{AilonCCL08} and studied in~\cite{AilonCCL08,SakS08,bgjjrw2010}
Property-preserving data reconstruction was introduced by Ailon, Chazelle, Comandur and Liu~\cite{AilonCCL08}. In this model, a reconstruction algorithm, called a {\em filter}, sits between a {\em client} and a {\em dataset}. A dataset is viewed as a function $f:\mathcal{D} \rightarrow \mathcal{R}$. Client accesses the dataset using {\em queries} of the form $x \in \mathcal{D}$ to the filter. The filter {\em looks up} a small number of values in the dataset and outputs $g(x)$, where $g$ must satisfy some fixed {\em structural} property (e.g., be monotone or have a low Lipschitz constant) and differ from $f$ as little as possible. Extending this notion, Saks and Seshadhri~\cite{SakS08} defined  {\em local} reconstruction.  A filter is {\em local} if it allows for a local (or distributed) implementation: namely, if the output function $g$ does not depend on the order of the queries.

Our results on TC-spanners are relevant to reconstruction of two properties of functions: monotonicity and having a low Lipshitz constant. Reconstruction of monotone functions was considered in~\cite{AilonCCL08,SakS08,bgjjrw2010}. A function $f: [m]^d \rightarrow \mathbb{R}$
is called {\em monotone} if $f(x) \leq f(y)$ for all $(x,y)\in E(\calh_{m,d})$. Reconstruction of functions with low Lipschitz constant was studied in \cite{lipschitz}. A function $f:[m]^d\to\mathbb{R}$ has Lipschitz constant $c$ if $|f(x)-f(y)|\leq c \cdot |x-y|_1$.
In \cite{bgjjrw2010}, the authors proved that the existence of a local filter for monotonicity of functions with low lookup complexity %per query
implies the existence of a sparse 2-TC-spanner of~$\calh_{m,d}$. In \cite{lipschitz}, an analogous connection is drawn between local reconstruction of functions with  low Lipschitz constant and 2-TC-spanners. Our improvement in the lower bound on the size of 2-TC-spanners of $\calh_{m,d}$ directly translates into improvement by the same factor in the lower bounds on lookup complexity of local filters for these two properties.

\paragraph{Key Management for Access Control Hierarchies.}
Atallah~\etal~\cite{AFB05} used \snote{Maybe say somewhere that they do not abstract out the combinatorial object they use} sparse Steiner TC-spanners to construct efficient key management schemes for
access control hierarchies.   An {\em access
hierarchy} is a partially ordered set $G$ of access classes.
%This is modeled by a directed graph $G$ whose nodes, labeled $1,\dots,n$,  are access classes and whose edges indicate the ordering.
Each user is entitled to
access a certain class and all classes reachable from the corresponding node in $G$.  One
approach to enforcing the access hierarchy is to use a key management
scheme of the following form \cite{ABF06,AFB05,SFM07,atallah-journal}.  Each edge $(i,j)$ has an associated public key $P(i,j)$, and each node $i$, an associated secret key $k_i$.  Only users with
the secret key for a node have the required permissions for the associated access class.  The public
and secret keys are designed so that there is an efficient algorithm $A$ which takes $k_i$ and
$P(i,j)$ and generates $k_j$, but for each $(i,j)$ in $G$, it is computationally hard to generate
$k_j$ without knowledge of $k_i.$  Thus, a user can efficiently generate the required keys to access a
descendant class, but not other classes.  The number of runs of algorithm $A$ needed to generate
a secret key $k_v$ from a secret key $k_u$ is equal to $d_G(u,v)$.  To speed this up, Atallah~\etal~\cite{atallah-journal} suggest adding edges and nodes to $G$ to increase
connectivity. To preserve the access hierarchy represented by $G$, the new graph $H$ must be a Steiner TC-spanner of $G$. The number of edges in $H$ corresponds to
the space complexity of the scheme, while the stretch $k$ of the spanner corresponds to the time
complexity.

We note that the time to find the path from $u$ to $v$ is also important in this application. In our upper bounds, this time is $O(d)$, which for small $d$ (e.g., constant) is likely to be much less than $2g(n)$ or $3g(n)$, where $g(n)$ is the time to run algorithm $A$. This is because algorithm $A$ involves the evaluation of a cryptographic hash function, which is expensive in practice and in theory\footnote{Any hash function which is secure against $\poly(n)$-time adversaries requires $g(n) \geq \polylog n$ evaluation time under existing number-theoretic assumptions.}.

\section{Definitions and Observations}\label{sec:prelim}
%Sofya: defined in a footnote
%For positive integer $m$, we denote $\{1,\dots, m\}$ by $[m]$.
For integers $j \geq i$, an interval $[i,j]$ refers to the set $\{i, i+1,\dots,j\}$.
Unless otherwise specified, logs are always base $2$, except for $\ln$ which is the natural logarithm.

Each DAG $G=(V,E)$ is equivalent to a poset with elements $V$ and partial order $\preceq$, where $x\preceq y$ if $y$ is reachable from $x$ in $G$. Elements $x$ and $y$ are {\em comparable} if $x \preceq y$ or $y\preceq x$, and {\em incomparable} otherwise.
We write $x \prec y$ if $x \preceq y$ and $x \neq y$.
The {\em hypergrid $\calh_{m,d}$ with dimension $d$ and side length $m$} was defined in the beginning of \secref{results}. Equivalently, it is the poset on
elements $[m]^d$ with the {\em dominance order}, defined
as follows: $x \preceq y$ for two elements
$x,y\in [m]^d$  iff $x_i \leq y_i$ for all $i\in [d]$.

A mapping from a poset $G$ to a poset $G'$ is called an {\em embedding} if it respects the partial
order, that is, all $x,y\in G$ are mapped to $x',y'\in G'$ such that $x \preceq_G y$ iff $x'
\preceq_{G'} y'$.
\begin{definition}[Poset dimension (\cite{DushnikMiller})]\label{def:poset-dimension}
Let $G$ be a poset with $n$ elements.  The {\em dimension} of $G$ is the smallest integer $d$ such that
$G$ can be embedded into the hypergrid ${\cal H}_{n,d}$.
\end{definition}
\noindent
Dushnik and Miller \cite{dush-miller} proved
that for any $m>1$, the hypergrid $\calh_{m,d}$ has dimension exactly
$d$.

\begin{fact}\label{fact:embedding-into-grid}
Each $d$-dimensional poset with $n$ elements can be embedded into a hypergrid ${\cal H}_{n,d}$, so that for all $i\in [d]$, the $i$th coordinates of images of all points are distinct.
\end{fact}

%%Sofya: used this example in the intro.
%\paragraph{Non-Steiner TC-spanners versus Steiner TC-spanners.}
%In \cite{tc-spanners-soda}, the authors showed that for planar posets (and, more generally, families
%of minor-free posets), there exist non-Steiner $k$-TC-spanners of near linear size. Such spanners do not exist for general low-dimensional posets, which motivates the study of Steiner $k$-TC-spanners. This is illustrated by the following example:
%
%\begin{example}
%Consider the complete bipartite graph $K_{n,n}$, containing $n^2$ edges, with all edges oriented away from the
%same part.  It is clear that any non-Steiner $k$-TC-spanner for $K_{n,n}$ has to have all the $n^2$
%edges.  However, $K_{n,n}$ is a poset of dimension $2$  and by our upper bounds, a Steiner $2$-TC
%spanner for it has size $O(n\log^2 n)$. Indeed, to see that $K_{n, n}$ is embeddable into a
%$[2n]\times [2n]$ grid, first denote by the vertex labeled $(0,0)$ the minimum grid element.
%Then each of the $n$ left vertices of $K_{n,n}$  is mapped to a distinct grid vertex in  the set of
%incomparable vertices $\{ (i, n-1-i): 0\leq i\leq n-1\}$, and similarly  each right vertex is mapped
%to a distinct vertex  in the set $\{(2n-1-i, i+n): 0\leq i\leq n-1 \}$. It is easy to see that this
%is a proper embedding. Actually, in this case an even better $O(n)$-sized Steiner $2$-TC spanner
%can be obtained by adding a single Steiner point at $(n,n)$ and adding new edges from the left
%vertices to it, and from it to the right vertices.
%\end{example}

\paragraph{Sparse Steiner 2-TC-spanners for $d$-dimensional posets.} %\label{sec:up-bds-ctdim-pos}
We give a simple construction of sparse Steiner 2-TC-spanners for $d$-dimensional posets.
For constant $d$, it matches the lower bound from Section~\ref{sec:grid-lower-bounds} up to a constant factor. Note that the construction itself works for arbitrary, not necessary constant, $d$.

\begin{theorem}\label{thm:2tc-ub}
Every $d$-dimensional poset $G$ on $n$ elements has a Steiner $2$-TC-spanner $H$
of size $O(n \log^d n)$ that can be constructed in time $O(dn \log^d n)$. Moreover, for all $x,y\in G$, where $x\prec y$, one can find a path in $H$ from $x$ to $y$ of length at most 2 in time $O(d)$.
\end{theorem}

\begin{myproof}
Consider an $n$-element poset $G$ embedded into the hypergrid $\calh_{n,d}$, so that for all $i\in [d]$, the $i$th coordinates of images of all points are distinct. (See Fact~\ref{fact:embedding-into-grid}).
 In this proof, assume that the hypergrid coordinates start with 0, i.e., its vertex set is $[0, n - 1]^d$. Let $\ell=\lceil \log n \rceil$ and $b(t)$ be the $\ell$-bit binary representation of $t$, possibly with leading zeros.
 Let $p_i(t)$ denote the $i$-bit prefix of $b(t)$ followed by a single 1 and then $\ell - i - 1$ zeros.
 Let $lcp(t_1, t_2)=p_i(t_1)$, where $i$ is the length of the longest common prefix of $b(t_1)$ and $b(t_2)$.

  To construct a Steiner $2$-TC-spanner $(V_H,E_H)$ of $G$, we insert at most $\ell^d$ edges into $E_H$ per each poset element. Consider a poset element with coordinates $x = (x_1, \dots, x_d)$ in the embedding.
  For each $d$-tuple $(i_1, \dots, i_d) \in [0,\ell - 1]^d$, let $p$ be a hypergrid vertex whose coordinates have binary representations $(p_{i_1}(x_1), \dots, p_{i_d}(x_d))$.
  If $x \prec p$, we add an edge $(x,p)$ to $E_H$; otherwise, if $p \prec x$ we add an edge $(p,x)$ to $E_H$. Note that only edges between comparable points are added to $E_H$.

We have that $E_H$ contains $O(n (\lceil \log n \rceil)^d)$ edges. If $d = O(\log n)$, then
$(\lceil \log n \rceil )^d \leq (\log n + 1)^d = (\log^d n)(1 + 1/\log n)^d = O(\log^d n)$, using the well-known inequality that
$1+x \leq e^{x}$. On the other hand, if $d = \Omega(\log n)$, the bound of this theorem holds trivially.
Hence, $E_H$ contains $O(n \log^d n)$ edges.\snote{With the current write up, we have to add ceilings to $\log n$ here; can we avoid it? No, O() does not take care of the ceilings because the whole thing is raised to power $d$}
 It can be constructed in $O(dn \log^d n)$ time, as described, if bit operations on coordinates can be performed in $O(1)$ time.
 \begin{comment}
  draw one edge from $x$ to a vertex on the grid with coordinates
  $(c_12^{i_1}, \dots, c_d2^{i_d})$, such that for all $k$ we have $x_k < c_k2^{i_k}$ and $c_k$ is minimal and another edge to a vertex with coordinates
  $(c'_12^{i_1}, \dots, c'_d2^{i_d})$, such that for all $k$ we have $x_k \ge c'_k2^{i_k}$ and $c'_k$ is maximal.
  If the other end of the edge falls outside of the considered grid $\calh_{n,d}$, we don't draw it.
  Clearly, there will be $O(n\log^d n)$ edges in this construction.
\end{comment}

   For all pairs of poset elements $x = (x_1, \dots, x_d)$ and $y = (y_1, \dots, y_d)$, such that $x \prec y$, there is an intermediate point
   $z$ with coordinates whose binary representations are $(lcp(x_1, y_1), \dots, lcp(x_d, y_d))$. By construction, both edges $(x, z)$ and $(z, y)$ are in $E_H$.
   Point $z$ can be found in $O(d)$ time, since $lcp(x_i, y_i)$ can be computed in $O(1)$ time, assuming $O(1)$ time bit operations on coordinates.
\begin{comment}
   we can find an intermediate point $z = (z_1, \dots, z_d)$, such that both edges $(x,z)$ and $(z,y)$ exist in the constructed $2$-TC-spanner.
   For each of $d$ coordinates we find $z_i$ separately in $O(1)$ time, assuming that we can do bit operations in $O(1)$ time.
  Suppose $x_i$ and $y_i$ are expressed as $a_1 \dots a_\ell$ and $b_1 \dots b_\ell$ in binary notation, then using bit operations we can find a binary number $z_i = c_1 \dots c_k10\dots0$ of length $\ell$, such that $c_1 \dots c_k$ is the longest common prefix of $a_1 \dots a_\ell$ and $b_1 \dots b_\ell$.\gnote{The case of $k = \ell$ should be considered separately.} Note that $z_i$ is the minimal number, divisible by $2^{\ell - k - 1}$, which is greater than $x_i$ and also maximal number, divisible by $2^{\ell - k - 1}$, which is less than or equal to $y_i$.
  Using this observation for every coordinate we conclude that both edges $(x, z)$ and $(z,y)$ are present in the graph.
\end{comment}
\end{myproof}

\paragraph{Equivalence of Steiner and non-Steiner TC-spanners for hypergrids.} Our lower bound on the size of $2$-TC-spanners for $d$-dimensional posets of size $n$ is obtained by
proving a lower bound on the size of the Steiner $2$-TC-spanner of $\calh_{m,d}$
where $m = n^{1/d}$.  The following lemma, used in Section~\ref{sec:>2-dim-poset-lb}, implies \corref{steiner-points-do-not-help}  that shows that sparsest Steiner
and non-Steiner $2$-TC-spanners of $\calh_{m,d}$ have the same size.

\begin{lemma}\label{lem:steiner-points-do-not-help}
Let $G$ be a poset on elements $V\subseteq[m]^d$ with the dominance order and $H=(V_H, E_H)$ be
a Steiner $k$-TC-spanner of $G$ with minimal $V_H$. Then $H$ can be embedded into ${\cal H}_{m,d}$.
\end{lemma}

\begin{myproof}
For each $s\in V_H-V$, we define $Prev(s)=\{x\in V:~x\prec s\}$.
If $Prev(s)=\varnothing$ then $V_H$ is not minimal because $H$ remains a Steiner
$k$-TC-spanner of $G$ when $s$ is removed.  We map each Steiner vertex $s$ to $r(s)$,
the replacement of $s$ in $[m]^d$, whose $i$th coordinates for all $i\in[d]$ are $\max_{x\in Prev(s)} x_i$.

Consider an edge $(x,y)$ in $G$.  If $x,y\in V$ our embedding does
not alter that edge.  If $x\in V$, $y\in V_H-V$ then $x\in Prev(y)$
and $x\prec r(y)$ by the definition of $r$.  If $x,y\in V_H-V$
then $Prev(x)\subseteq Prev(y)$ and the monotonicity of $\max(S)$ for sets
implies $r(x)\preceq r(y)$.  Finally, if $x\in V_H-V$ and $y\in V$
then for each $z\in Prev(x)$ and each $i\in[d]$, we have
$z_i\le y_i$ because $z\prec x\prec y$, and this implies
$r(x)\preceq y$.
\end{myproof}

\begin{corollary}
\label{cor:steiner-points-do-not-help}
If ${\cal H}_{m,d}$ has a Steiner $k$-TC-spanner $H$, it also
has a $k$-TC-spanner with the same number of nodes and at most the same number of edges.
\end{corollary}

\section{Our Lower Bound for 2-TC-spanners of the Hypergrid}\label{sec:grid-lower-bounds}
In this section, we prove Theorem~\ref{thm:grid-lb} that gives a nearly tight lower bound on the size of (Steiner) 2-TC-spanners of the hypergrids $\calh_{m,d}$.  By Corollary~\ref{cor:steiner-points-do-not-help}, we only have to consider non-Steiner TC-spanners.

\begin{proofof}{\thmref{grid-lb}}
We start by introducing a linear program for the sparsest 2-TC-spanner of an arbitrary graph. Our lower bound on the size of a 2-TC-spanner of $\calh_{m,d}$ is obtained by finding a feasible solution to the dual program, which, by definition, gives  a lower bound on the objective function of the primal.

\paragraph{Integer linear program for sparsest 2-TC-spanner.}\label{sec:linear-program}
%We write $x \prec y$ if there is a path from $x$ to $y$ in the directed graph $G$ and $x \preceq y$ if either $x \prec y$ or $x = y$.

  For every graph, we can find the size of a sparsest 2-TC-spanner by solving the following \{0,1\}-linear program which is a special case of a more general program from \cite{tc-spanners-soda} for directed $k$-spanners.
  For all vertices $u$, $v$ satisfying $u \preceq v$, we introduce variables $x_{uv} \in \{0,1\}$.
  If $H=(V, E_H)$ is the corresponding 2-TC-spanner, $x_{uv}=1$ iff $(u,v) \in E_H$.
  For all vertices $u$, $v$, $w$ satisfying $u \preceq w \preceq v$, we introduce auxiliary variables $x'_{uwv} \in \{0,1\}$.
  If $H=(V, E_H)$ is the corresponding 2-TC-spanner, $x'_{uwv} = 1$ if both $(u,w)$ and $(w,v)$ are in $E_H$.
  The \{0,1\}-linear program is as follows:
    \begin{align*}
    \mbox{minimize \ \ } & \sum\limits_{u,v \colon u \preceq v} x_{uv}& \\
    \mbox{subject to \ \ } & x_{uw} - x'_{uwv} \ge 0, x_{wv} - x'_{uwv} \ge 0 &\forall u, v, w \colon u \preceq w \preceq v; \\
    &\sum\limits_{w \colon u \preceq w \preceq v} x'_{uwv} \ge 1 &\forall u, v \colon u \preceq v; \\
     &x_{uv} \in \{0,1\}&\forall u, v \colon u \preceq v;\\
     &x'_{uwv} \in \{0,1\}&\forall u, v, w \colon u \preceq w \preceq v.
    \end{align*}
  The size of the sparsest $2$-TC-spanner and the optimal value of the objective function of this linear program differ by $\sum\limits_u x_{uu} \leq m^d$. \ignore{\snote{Isn't this sum exactly $n$? Maybe we should say this instead of a factor of 2? o.w., it is not clear the factor of 2 between which quantities?}}
  Since we are considering asymptotic behavior of the size of the $2$-TC-spanner, this difference can be ignored.

  Every feasible solution of the following fractional relaxation of a dual linear program gives a lower bound on the objective function of the primal.
  \begin{align}
    \mbox{maximize \ \ } & \sum\limits_{u,v \colon u \preceq v} y_{uv}& \nonumber\\
    \mbox{subject to \ \ } &\label{dualA} \sum\limits_{w \colon v \preceq w} y'_{uvw} + \sum\limits_{w \colon w \preceq u} y''_{wuv} \le 1& \forall u, v \colon u \preceq v;\\
    &\label{dualB}y_{uv} - y'_{uwv} - y''_{uwv} \le 0& \forall u, v, w \colon u \preceq w \preceq v; \\
    &y_{uv} \ge 0& \forall u \preceq v ;\nonumber\\
    &y'_{uwv} \ge 0, y''_{uwv} \ge 0& \forall u \preceq w \preceq v. \nonumber
  \end{align}

\begin{comment}
  In \cite{tc-spanners} it is shown that sparsest 2-TC-spanner problem for general graphs is $\Omega(\log n)$-inapproximable, unless P=NP.
  The upper bound on integrality gap of the presented linear program is $O(\log n)$, as shown in \cite{DinitzPhD}, and thus matches this lower bound.
  However, when restricted to specific family of graphs $\mathcal{H}_{m,d}$, integrality gap turns out to be $O(1)$ and this lets us get tight lower bound on $S_2(\mathcal{H}_{m,d})$.
\end{comment}

%\subsection{Lower bound for 2-TC-spanner of ${\cal H}_{m,d}$}\label{sec:grid-lower-bound-lp}
\paragraph{Finding a feasible solution for the dual.}
The rest of the proof of the \thmref{grid-lb} can be broken down into the following steps:
\begin{enumerate}
\item We choose initial values $\hat y_{uv}$ for the variables $y_{uv}$ of the dual program and, in \lemref{sum-of-constraints}, give a lower bound on the resulting value of the objective function of the primal program.

\item We choose initial values $\hat y'_{uvw}$ and $\hat y''_{uvw}$ for variables $y'_{uvw}$ and $y''_{uvw}$ to ensure that (\ref{dualB}) holds.

\item In \lemref{lemma2}, we give an upper bound on the left side of (\ref{dualA}) for all $u \preceq v$.
  Our bound is a constant larger than 1 and independent of $n$.
  We obtain a feasible solution to the dual by dividing the initial variable values (and, consequently, the value of objective function) by this constant.
\end{enumerate}

\noindent \textbf{Step 1.}  For a vector $x = (x_1, \dots, x_d) \in \mathbb [0, m - 1]^d$, let $\Vol(x)$ denote $\prod_{i \in [d]} (x_i + 1)$. This corresponds to the number of hypergrid points inside a $d$-dimensional box with corners $u$ and $v$, where $v-u=x$. To obtain the desired lower bound, we set $\hat y_{uv} = \frac{1}{\Vol(v - u)}$ for all $u \preceq v$.    This gives the value of the objective function of the dual program, according to the following lemma.
\begin{lemma}\label{lem:sum-of-constraints}
\ignore{For $m > e^{2d}$, $\displaystyle \sum\limits_{u, v \colon u \preceq v} \hat y_{uv} > \frac{m^d \ln^d m}{2}.$}
$\displaystyle \sum\limits_{u,v \colon u \preceq v} \hat{y}_{uv} > m^d (\ln m - 1)^d$
\end{lemma}
\begin{myproof}
Substituting $1/(\Vol(v - u))$ for $\hat y_{uv}$, we get:
\begin{eqnarray*}
\sum\limits_{u,v \colon u \preceq v} \hat y_{uv}
&=&
\sum\limits_{u, v \colon u \preceq v} \frac{1}{\Vol(v - u)} =
\sum\limits_{l \in [m]^d}
\prod\limits_{i \in [d]} \frac{m - l_i + 1}{l_i} = \left(\sum\limits_{l \in [m]} \frac{m - l + 1}{l}\right)^d  \\
&>& ((m + 1)\ln (m + 1) - m)^d > m^d (\ln m - 1)^d.
\end{eqnarray*}
\ignore{Since $(\ln m - 1)^d \ge \frac{\ln^d m}{2}$ for all $m > e^{2d}$, the statement of the lemma follows.}
\end{myproof}

\noindent \textbf{Step 2.}  The values of $\hat y'_{uvw}$ and $\hat y''_{uvw}$ are set as follows to satisfy (\ref{dualB}) tightly (without any slack):
\begin{eqnarray*}
  \hat y'_{uvw} = \hat y_{uw}\frac{\Vol(v - u)}{\Vol(v - u) + \Vol(w - v)},~
  \hat y''_{uvw} = \hat y_{uw} - \hat y'_{uvw} = \hat y_{uw}\frac{\Vol(w - v)}{\Vol(v - u) + \Vol(w - v)}.
\end{eqnarray*}

\noindent \textbf{Step 3.}  The initial values $\hat y'_{uvw}$ and $\hat y''_{uvw}$ do not necessarily satisfy (\ref{dualA}).
Next, we give the same upper bound on the left hand side of all constraints (\ref{dualA}).
\begin{lemma}\label{lem:lemma2}
    For all $u \preceq v$, $\sum\limits_{w \colon v \preceq w} \hat y'_{uvw} + \sum\limits_{w \colon w \preceq u} \hat y''_{wuv} \le (4\pi)^d$.
\end{lemma}
\begin{myproof}
Below we denote $v - u$ by $x^0 = (x^0_1, \dots, x^0_d)$, a $d$-dimensional vector of ones (1, \dots, 1) as $\fone$ and $\prod_{i \in [d]} dx_i$ by $dx$.
\begin{eqnarray}
    \sum\limits_{w \colon v \preceq w} \hat y'_{uvw} + \sum\limits_{w \colon w \preceq u} \hat y''_{wuv}
    &=&  \sum\limits_{w \colon v \preceq w} \hat y_{uw} \frac{\Vol(v - u)}{\Vol(v - u) + \Vol(w - v)} +
\sum\limits_{w \colon w \preceq u} \hat y_{wv} \frac{\Vol(v - u)}{\Vol(u - w) + \Vol(v - u)} \nonumber \\
 &<& 2 \sum\limits_{x \in [0,m]^d}
 \frac{\Vol(x^0)}{V(x^0 + x)(\Vol(x^0) + \Vol(x))} \nonumber \\
 &\le& 2^{2d + 1} \sum\limits_{x \in [1,m + 1]^d}
 \frac{\Vol(x^0)}{V(x^0 + x)(\Vol(x^0) + \Vol(x))} \label{eqn:eq1}\\
 &<& 2^{2d + 1} \int_{\mathbb{R}^{d}_{+}} \frac{V(x^0) d x}{\Vol(x^0 + x)(\Vol(x^0) + \Vol(x))} \label{eqn:eq2}\\
 &=& 2^{2d + 1} \int_{\mathbb{R}^{d}_{+}} \frac{\Vol(x^0) d t}{\Vol(t)(\Vol(x^0) + \prod\limits_{i}(t_i (x^0_i + 1) + 1))} \label{eqn:eq3}\\
 &<& 2^{2d + 1} \int_{\mathbb{R}^{d}_{+}} \frac{\Vol(x^0) d t}{\Vol(t)(\Vol(x^0) + \prod\limits_{i}t_i (x^0_i + 1))} \nonumber \\
 &=& 2^{2d + 1} \int_{\mathbb{R}^{d}_{+}} \frac{d t} {V(t)(\fone + V(t - 1))}. \nonumber
\end{eqnarray}

\gnote{Added explanations here.}
The first equality above is obtained by plugging in values of $\hat y'$ and $\hat y''$ from Step 2 with appropriate indices.
The first inequality is obtained by extending each sum to the whole subgrid.
Here (\ref{eqn:eq1}) holds because $\frac{1}{\Vol(u)} \le \frac{2^d}{\Vol(u + 1)}$ for all $u$, such that $u_i \ge 0$.
In (\ref{eqn:eq2}), the sum can be bounded from above by the integral because the summand is monotone in all variables.
To get (\ref{eqn:eq3}), we substitute $x$ by $t$, which satisfies $x_i = t_i(x^0_i + 1)$.
In the last inequality, we substitute $\Vol(x^0)$ for $\prod\limits_{i}(x^0_i + 1)$.

\begin{claim}\label{claim:intestimate}
Let $I_d= \int_{\mathbb{R}^{d}_{+}} \frac{d t} {V(t)(\fone + V(t - 1))}$. Then $I_d \le \frac{\pi^d}{2}$ for all $d$.
\end{claim}
Lemma~\ref{lem:lemma2} follows from Claim~\ref{claim:intestimate} whose proof is deferred to Appendix~\ref{app:grid-lb-integral}.
\end{myproof}

Finally, we obtain a feasible solution by dividing initial values $\hat y_{uv}$, $\hat y'_{uvw}$ and $\hat y''_{uvw}$ by the upper bound $(4\pi)^d$ from \lemref{lemma2}. Then \lemref{sum-of-constraints} gives the desired bound on the value of the objective function.
\begin{equation*}
\sum_{u,v \colon u \preceq v} \frac{\hat y_{uv}}{(4\pi)^d} > m^d \left(\frac{\ln m - 1}{4\pi}\right)^d
\end{equation*}
This completes the proof of \thmref{grid-lb}.\end{proofof}

\newcommand{\Bx}{{\mathbb B}}
\newcommand{\Ex}{{\mathbb E}}
\newcommand{\Px}{{\mathbb P}}
\newcommand{\BxP}{{{\mathbb B \mathbb P}}}
\newcommand{\ivec}{\vec{\imath}}
\newcommand{\jvec}{\vec{\jmath}}
\newcommand{\cJ}{{\mathcal J}}
\newcommand{\cP}{{\mathcal G}}
\newcommand{\cG}{{\mathcal G}}
\newcommand{\valph}{\vec{\alpha}_{\ivec}}
\newcommand{\vlamb}{\vec{\lambda}_{\ivec}}
\newcommand{\vkapp}{\vec{\kappa}}
\newcommand{\E}{\mathop{\mathbb E}}
\newcommand{\pairind}[1]{\lambda_i{(#1)} - \pi_i{(#1)}}
\newcommand{\boxof}{\lambda_i}
\newcommand{\topof}{\pi_i}
\newcommand{\boxofvec}{\lambda_{\ivec}}
\newcommand{\topofvec}{\pi_{\ivec}}
\newcommand{\groups}{g}
\newcommand{\lprime}{\ell'}

\section{Our Lower Bound for $k$-TC-spanners of $d$-dimensional Posets for $k>2$}\label{sec:lb-posetk>2}
In this section, we prove Theorem~\ref{thm:lb-steiner-2dim-k>2}.

\begin{proofof}{Theorem~\ref{thm:lb-steiner-2dim-k>2}}
Unlike in the previous section, the poset which attains the lower bound is constructed probabilistically, not explicitly.  Let $\calg_d$ be a distribution on $n$-element posets embedded in $\calh_{n,d}$, where all poset elements differ on coordinates in dimension 1, and for each such coordinate $a \in [n]$, an element $p_a$ is chosen uniformly and independently  from $\{a\} \times [n]^{d - 1}$. The partial order is then given by the dominance order $x \preceq y$ on $\calh_{n,d}$.

Recall that $S_k(G)$ denotes the size of the sparsest Steiner $k$-TC-spanner of poset $G$. The following lemma gives a lower bound on the expected size of a Steiner $k$-TC-spanner of a poset drawn from $\calg_d$.

\begin{lemma}\label{lem:lb-steiner-d>=2}
 $\E\limits_{G \gets \cG_d}[S_k(G)]=\Omega(n \log^{\lceil (d - 1) / k \rceil} n)$ for all $k \ge 3$ and constant $d \ge 2$.
\end{lemma}

To simplify the presentation, we first prove the special case of Lemma~\ref{lem:lb-steiner-d>=2} for 2-dimensional posets in Section~\ref{sec:2-dim-poset-lb}. The general case is proved in Section~\ref{sec:>2-dim-poset-lb}.
Since Lemma~\ref{lem:lb-steiner-d>=2} implies the existence of a poset $G$, for which every Steiner $k$-TC-spanner
has $\Omega(n \log^{\lceil (d - 1) / k \rceil} n)$ edges,
Theorem~\ref{thm:lb-steiner-2dim-k>2} follows.
\end{proofof}

\subsection{The case of $d=2$}\label{sec:2-dim-poset-lb}
This section proves a special case of Lemma~\ref{lem:lb-steiner-d>=2} for 2-dimensional posets, which illustrates many of the ideas used in the proof of the general lemma. In both proofs, we assume that $\ell=\log n$ is an integer.
\begin{lemma}[Special case of Lemma~\ref{lem:lb-steiner-d>=2}]\label{lem:lb-steiner-d=2}
$\E\limits_{G \gets \cG_2}[S_k(G)]=\Omega(n \log n)$ for all $k \ge 3$ and $d = 2$.
\end{lemma}
%\begin{myproof}%{\lemref{lb-steiner-d=2}}
%% Sofya: couldn't make the proof environment work with wrapfigure
\noindent{\em Proof.}
To analyze the expected number of edges in a Steiner TC-spanner, we consider $\ell$ partitions of $[n]^2$ into strips. We call strips {\em boxes} for compatibility with the case of general $d$.

\begin{definition}[Box partition]\label{def:box-partition-d=2}
For each $i\in[\ell]$, we define sets of equal size that partition $[n]$ into $2^i$ intervals:
the $j$th such set, for $j\in [2^i]$, is $I^i_j= [(j-1)2^{\ell-i}+1, j2^{\ell-i}]$.
Given $i\in[\ell]$, and $j\in[2^{i}]$, the {\em box}
$\Bx(i,j)$ is $[n]\times I^{i}_{j}$ and the {\em box partition}  $\BxP(i)$ is a partition of $[n]^2$ that contains boxes $\Bx(i,j)$ for all $j\in[2^{i}]$.
\end{definition}

\begin{wrapfigure}[9]{r}{2in}
\vspace{-30pt}
\begin{center}
%\ifnum\llncs=1
%\includegraphics[width=\columnwidth]{steiner-illustration_2dim-random-poset}
%\else
\includegraphics[width=2in]{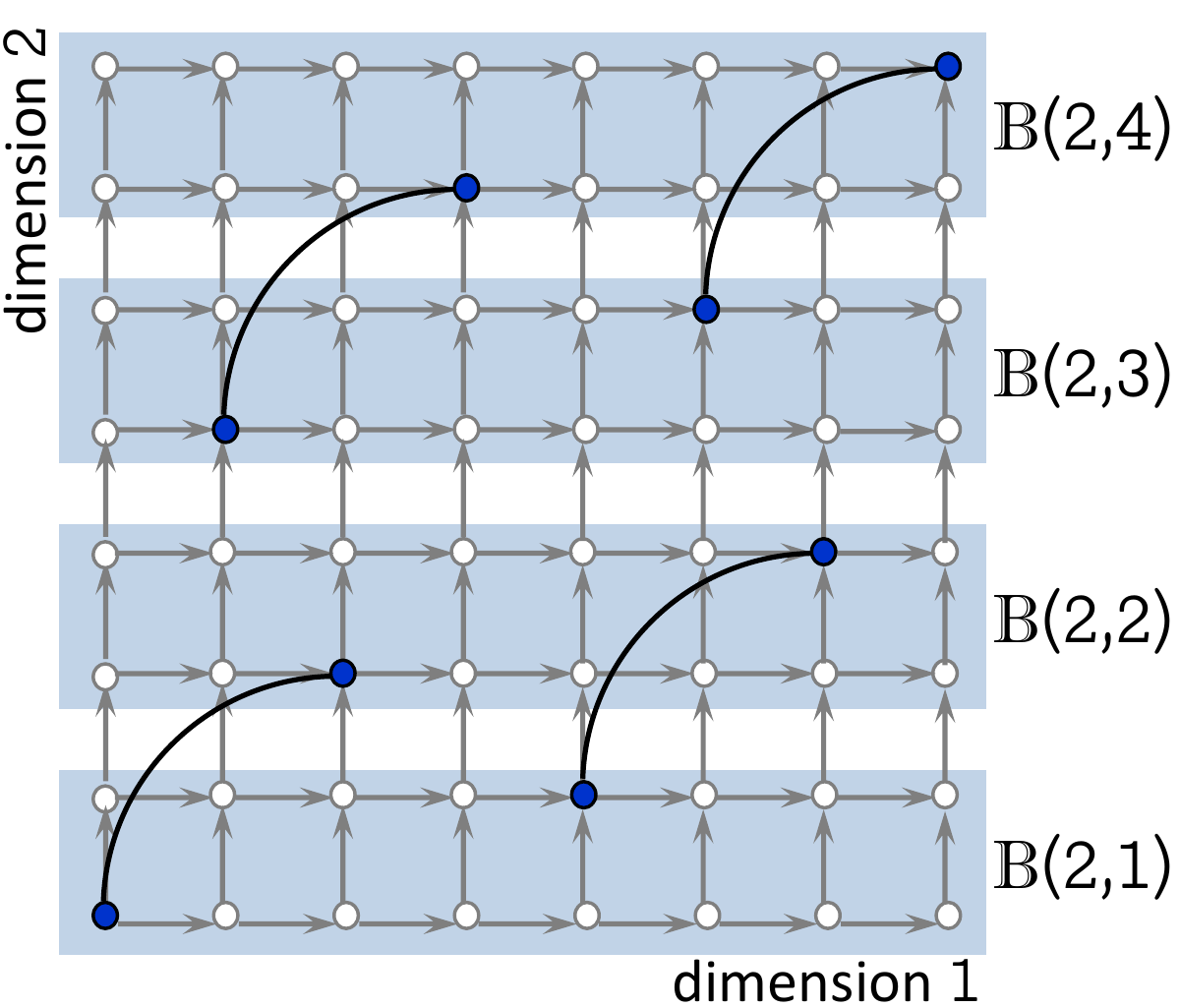}
%\fi
\caption{Box partition $\BxP(2)$ and jumps it generates.}
%\label{fig:}
\end{center}
\end{wrapfigure}
We analyze the expected number of edges that cross from boxes with an odd index $j$ into boxes with index $j+1$ with respect to partition $\BxP(i)$ for all $i\in[\ell]$. To do that, we identify pairs of poset elements that force such edges to appear. The pairs of their first coordinates are called {\em jumps} and are defined next.
\begin{definition}[Jumps]
A {\em jump generated by the partition} $\BxP(i)$ is a pair $(a,b)$ of coordinates\snote{Why not a pair of nodes? It would be easier to discuss it.} in dimension 1, such that for %some $i\in[\ell]$ and
some odd $j\in[2^i]$,
the following holds:
$p_a\in \Bx(i,j)$,
$p_b\in \Bx(i,j+1)$,
while $p_c \notin \Bx(i,j)\cup \Bx(i,j+1)$ for all $c\in(a,b)$. Let $\cJ$ denote the set of jumps generated by all partitions $\BxP(i)$ for $i\in[\ell]$.
\end{definition}

We use two properties of $\cJ$, given in Claims~\ref{claim:mapping-of-jumps-dim2} and~\ref{claim:expected-jumps-dim2}.

\begin{claim}\label{claim:mapping-of-jumps-dim2}
Let $G$ be a poset, embedded into $\calh_{n,2}$, and $H=(V_H,E_H)$ be
 a Steiner $k$-TC-spanner of $G$. Then there exists a 1-1
mapping from $\cJ$ to $E_H$.
\end{claim}
\begin{proof}
By \lemref{steiner-points-do-not-help}, we can assume that all Steiner vertices of $H$ are embedded into $\calh_{n,2}$. Given a jump $(a,b)$, we define $e(a,b)\in E_H$ by following a path from $p_a$ to $p_b$ in
$H$. This path is contained in $\Bx(i,j)\cup \Bx(i,j+1)$, and $e(a,b)$ is defined as the edge
on that path
that starts in
$\Bx(i,j)$  and ends in
$\Bx(i,j+1)$.

To show that $e(a,b)$ is a 1-1 mapping, we describe an inverse mapping.
To determine $(a,b)$ from $e(a,b) = ((u_1, u_2),(v_1, v_2))$ we find a
number in $[u_2, v_2 - 1]$, which is divisible by the largest power of 2 and so
has a form $j2^{\ell - i}$, from which we determine $i$ and $j$.
Among all jumps $(a',b')$ defined by boxes $\Bx(i,j)$, $\Bx(i, j + 1)$ only one
can satisfy $a' \le u_1 \le v_1 \le b'$.
\end{proof}

\begin{claim}\label{claim:expected-jumps-dim2}
When a poset $G$ is drawn from the distribution $\calg_2$, the expected size of $\cJ$ is at least $n(\ell-1)/4$.
\end{claim}
\begin{proof}
We first find the expected number of jumps generated by the partition $\BxP(i)$.
%For $u \in [n]^2$, we define $\boxof(u)$ as such $j$ that $u \in \Bx(i,j)$.
We group boxes $\Bx(i,j)$ and $\Bx(i,j+1)$ for odd $j$ into box pairs.
For $u \in [n]^d$, we define location $\boxof(u)$ as such $j$ that $u \in \Bx(i, j)$
and parity $\topof(u) = (\boxof(u) + 1)$ mod $2$.
Importantly, random variables $\topof(p_a)$ are independent and uniform over $\{0,1\}$ for all $a \in [n]$.

We group together elements $p_a$ that have equal values of
$\pairind{p_a}$, and sort elements within groups in increasing order of their first coordinate $a$. Observe that random variables $\topof(p_a)$ within each group are uniform and independent because random variables $\pairind{p_a}$ and $\topof(p_a)$ are independent for all $a$. Now, if we list $\topof(p_a)$ in the sorted order for all elements in a particular group, we get a sequence of $0$s and $1$s. Two consecutive entries correspond to a jump iff they are $01$. The last position in a group cannot correspond to the beginning of a jump. The number of positions that can correspond to the beginning of a jump in all groups is $n$ minus the number of nonempty groups, which gives at least $n-2^{i-1}$. For each such position, the probability that it starts a jump (i.e., the probability of $01$) is 1/4. Thus, the expected number of jumps generated by the partition $\BxP(i)$ is at least $(n - 2^{i - 1})/4$.

Summing over all $i \in [\ell]$, we get the expected number of jumps in all partitions:
$(n\ell-\sum_{i=1}^\ell 2^{i-1})/4> n(\ell-1)/4=\Omega(n\log n)$.
\end{proof}
Claims~\ref{claim:mapping-of-jumps-dim2} and~\ref{claim:expected-jumps-dim2} imply that, for a poset $G$ drawn from $\calg_2$, the expected number of edges in a Steiner TC-spanner $H$ of $G$ is $\Omega(n\log n)$, concluding the proof of \lemref{lb-steiner-d=2}.
\qed %\end{myproof}

\subsection{The case of $d > 2$}\label{sec:>2-dim-poset-lb}
\begin{myproof}[Proof of \lemref{lb-steiner-d>=2}]

Generalizing the proof for $d=2$, we consider $\ell^{d-1}$ partitions of $[n]^d$ into boxes, where $\ell=\log n$. In this proof, let $\lprime=\lfloor \ell/(d-1)\rfloor$ and $d'=\lceil (d-1)/k \rceil$.

\begin{definition}[Box partition]\label{def:box-partition-d>2}
Given vectors $\ivec=(i_1,\ldots,i_{d-1})\in[\lprime]^{d-1}$ and $\jvec=(j_1,\ldots,j_{d-1})\in[2^{i_1}]\times\cdots\times[2^{i_{d-1}}]$, the {\em box}
$\Bx(\ivec,\jvec)$ is $[n]\times I^{i_1}_{j_1}\times\ldots\times I^{i_{d-1}}_{j_{d-1}}$, and the {\em box partition}  $\BxP(\ivec)$ is a partition of $[n]^d$ that contains boxes $\Bx(\ivec,\jvec)$ for all eligible $\jvec$.
\end{definition}

To generalize the definition of the set of jumps $\cJ$, we denote $(d-1)$-dimensional vectors $(0, \dots, 0)$ and $(1,\ldots,1)$ by by $\fzero$ and $\fone$, respectively. We say that a vector $\jvec$ is {\em odd} if all of its coordinates are odd.

\begin{definition}[Jumps]
A {\em jump} generated by a box partition $\BxP(\ivec)$ is a pair $(a,b)$ of coordinates in dimension 1, such that for some vector $\ivec\in[\lprime]^{d-1}$ and some odd vector $\jvec$,
the following holds:
$p_a\in \Bx(\ivec,\jvec)$,
$p_b\in \Bx(\ivec,\jvec +\fone)$,
while $p_c \notin \Bx(\ivec,\jvec)\cup \Bx(\ivec,\jvec+\fone)$ for all $c\in(a,b)$. The set of jumps generated by all partitions $\BxP(\ivec)$ is denoted by $\cJ$.
\end{definition}

For $u \in [n]^d$, we define location $\boxofvec(u)$ as such $\jvec$ that $u \in \Bx(\ivec, \jvec)$
and parity $\topofvec(u) = (\boxofvec(u) + \fone)$ mod $2$, where mod is taken on each component of the vector.

\begin{comment}
Let $\boxofvec(u) = \jvec$\ when $u\in\Bx(\ivec,\jvec)$. For point $u\in [n]^d$ we
define 0-1 vector $\valph(u)$ by the
following property: $\vlamb(u)-\valph(u)$ is odd.
\end{comment}

\begin{claim}\label{claim:mapping-of-jumps-dim>2}
Let $G$ be a poset, embedded into $\calh_{n,d}$, and $H=(V_H,E_H)$ be
 a Steiner $k$-TC-spanner of $G$. Then there exists a mapping from
$\cJ$ to $E_H$ that maps $O(\ell^{d-1-d'})$ jumps to one edge.
\end{claim}
\begin{proof}
By \lemref{steiner-points-do-not-help}, we can assume that all Steiner vertices of $H$ are embedded into $\calh_{n,d}$.

First, we describe how to map a jump $(a,b)$ to an edge $e(a,b)\in E_H$. Each jump $(a,b)$ is generated by a box partition $\BxP(\ivec)$ for some $\ivec$.  We follow a path of length at most $k$ in
$H$ from $p_a$ to $p_b$, say, $(p_a=u_0,\ldots,u_k=p_b)$, and let $e(a,b)$ be
an edge on this path that maximizes the Hamming
distance between $\topofvec(u_c)$ and $\topofvec(u_{c+1})$. Note that this distance is
at least $d'$ because $\topofvec(u_0)=\fzero$ and $\topofvec(u_k)=\fone$.

Now we count the jumps mapped to an edge $e=(u,v)$.  First, we find all such jumps generated
by a single box partition $\BxP(\ivec)$.
They are defined by the pair of boxes \gnote{Can this be explained better?}
$\Bx(\ivec,\boxofvec(u)-\topofvec(u))$ and
$\Bx(\ivec,\boxofvec(u)-\topofvec(u)+\fone)$.
Then $[u_1,v_1]$ must be included in
one of the intervals $[a,b]$ defined by the jumps of this pair of boxes,
and those intervals are disjoint.  Hence, there is at most one such jump.

It remains to count box partitions $\BxP(\ivec)$ which can generate a jump mapped to a specific edge $e$.
A necessary condition is that $\boxofvec(v)-\boxofvec(u)$ is
a vector in $\{0,1\}^{d-1}$ with at least $d'$ 1's. There are less than $2^{d-1}$ such
vectors.  Consider one of these vectors, say, $\vec{\gamma}$.
If for some $t \in [d - 1]$, $\gamma_t=1$ then $i_t$ is uniquely determined by the largest power of 2
that divides a number in $[u_t,v_t-1]$.
When $\gamma_t=0$, there are at most
$\lprime$ possible values of $i_t$ because $\ivec\in[\lprime]^{d-1}$. Since $d$ is a constant, there are at most $2^{d-1}(\lprime)^{d-1-d'}= O(\ell^{d-1-d'})$ possible vectors $\ivec$, such that $\BxP(\ivec)$ could have generated a jump $(a,b)$.

Therefore,  $O(\ell^{d-1-d'})$ jumps map to the same edge of $E_H$.
\end{proof}

\begin{claim}\label{claim:expected-jumps-dim>2}
When a poset $G$ is drawn from the distribution $\calg_d$, the expected size of $\cJ$ is $\Omega(\ell^{d-1}n)$.
\end{claim}
\begin{proof}
To find the expected number of jumps generated by $\BxP(\ivec)$,
we analyze the sequence $\topofvec(p_a)$, $a\in[n]$.  The values in that
sequence are independent and uniformly distributed over $\{0,1\}^{d-1}$. First, we remove
all values different from $\fzero$ and $\fone$, and obtain a sequence of
expected length $n/2^{d-2}$. Then, in that sequence, we group together elements $p_a$ with equal values of $\boxofvec(p_a)-\topofvec(p_a)$,
and sort elements within groups in increasing order of their first coordinate $a$.
Observe that random variables $\topofvec(p_a)$ within each group are uniform and independent because random variables
$\boxofvec(p_a)-\topofvec(p_a)$ and $\topofvec(p_a)$ are independent for all $a$.
Now, if we list $\topofvec(p_a)$ in the sorted order for all elements in particular group, we get a sequence of $\fzero$s and $\fone$s.
Two consecutive entries correspond to a jump iff they are $\fzero\,\fone$.

Let $\groups(\ivec)$ denote the number of groups, that is, the number of possible values of $\boxofvec(p_a)-\topofvec(p_a)$. Then $\groups(\ivec) = \prod_{t=1}^{d-1}2^{i_t-1}$. Summing over all box partitions, we get:
$$
\sum_{\ivec\in[\lprime]^{d-1}}\groups(\ivec)=
\sum_{\ivec\in[\lprime]^{d-1}}\prod_{t=1}^{d-1}2^{i_t-1}=
\left(\sum_{t=1}^{\lprime}2^{t-1}\right)^{d-1}<2^{\lprime(d-1)}\le 2^\ell = n.
$$

On every position in the reordered sequence that is not the final position in its group,
the expected number of jumps started is 1/4, so the expected number of
jumps is at least $(n/2^{d-2}-\groups(\ivec))/4 = n/2^d-\groups(\ivec)/4$.
Therefore, the expected number of jumps generated by all box partitions is at least
$$(\lprime)^{d-1}n/2^d -\frac{1}{4}\sum\limits_{\ivec\in[\lprime]^{d-1}} \groups(\ivec)\ge
(\lprime)^{d-1}n/2^d -n/4=\Omega(\ell^{d-1}n).$$
The last equality holds because $d$ is constant.
\end{proof}

Claim~\ref{claim:expected-jumps-dim>2} gives a lower bound of $\Omega(\ell^{d-1}n)$ on the expected number of jumps in a poset $G$.
The mapping from Claim~\ref{claim:mapping-of-jumps-dim>2} takes $O(\ell^{d-1-d'})$ of these jumps to one edge. Thus, the expected
 number of edges in a Steiner TC-spanner $H$ of $G$ is $\Omega(n\ell^{d'})=\Omega(n \log^{\lceil (d - 1) / k \rceil} n)$. This concludes the proof of \lemref{lb-steiner-d>=2}.
\end{myproof}

{\small
\bibliography{spanners-bibliography} % BibLIOGRAPHY

\begin{thebibliography}{10}

\bibitem{a28}
W.~Ackermann.
\newblock Zum {Hilbertshen} aufbau der reelen zahlen.
\newblock {\em Math. Ann.}, 99:118--133, 1928.

\bibitem{AilonCCL08}
N.~Ailon, B.~Chazelle, S.~Comandur, and D.~Liu.
\newblock Property-preserving data reconstruction.
\newblock {\em Algorithmica}, 51(2):160--182, 2008.

\bibitem{AlonSchieber87}
N.~Alon and B.~Schieber.
\newblock Optimal preprocessing for answering on-line product queries.
\newblock Technical Report 71/87, Tel-Aviv University, 1987.

\bibitem{atallah-journal}
M.~J. Atallah, M.~Blanton, N.~Fazio, and K.~B. Frikken.
\newblock Dynamic and efficient key management for access hierarchies.
\newblock {\em ACM Trans. Inf. Syst. Secur.}, 12(3):1--43, 2009.

\bibitem{ABF06}
M.~J. Atallah, M.~Blanton, and K.~B. Frikken.
\newblock Key management for non-tree access hierarchies.
\newblock In {\em SACMAT}, pages 11--18, 2006.

\bibitem{AFB05}
M.~J. Atallah, K.~B. Frikken, N.~Fazio, and M.~Blanton.
\newblock Dynamic and efficient key management for access hierarchies.
\newblock In {\em ACM Conference on Computer and Communications Security},
  pages 190--202, 2005.

\bibitem{bgjjrw2010}
A.~Bhattacharyya, E.~Grigorescu, M.~Jha, K.~Jung, S.~Raskhodnikova, and
  D.~Woodruff.
\newblock Lower bounds for local monotonicity reconstruction from
  transitive-closure spanners.
\newblock In {\em Proceedings of the 14th RANDOM}, pages 448--461, 2010.

\bibitem{tc-spanners-soda}
A.~Bhattacharyya, E.~Grigorescu, K.~Jung, S.~Raskhodnikova, and D.~P. Woodruff.
\newblock Transitive-closure spanners.
\newblock In {\em Proceedings of the Twentieth Annual ACM-SIAM Symposium on
  Discrete Algorithms (SODA)}, pages 932--941, 2009.

\bibitem{BTS94}
H.~L. Bodlaender, G.~Tel, and N.~Santoro.
\newblock Trade-offs in non-reversing diameter.
\newblock {\em Nordic J. of Computing}, 1(1):111--134, 1994.

\bibitem{CFL83b}
A.~K. Chandra, S.~Fortune, and R.~J. Lipton.
\newblock Lower bounds for constant depth circuits for prefix problems.
\newblock In {\em Proc. 10th Annual International Conference on Automata,
  Languages, and Programming}, pages 109--117, 1983.

\bibitem{CFL83a}
A.~K. Chandra, S.~Fortune, and R.~J. Lipton.
\newblock Unbounded fan-in circuits and associative functions.
\newblock In {\em Proc.\ 15th Annual ACM Symposium on the Theory of Computing},
  pages 52--60, 1983.

\bibitem{cha87}
B.~Chazelle.
\newblock Computing on a free tree via complexity-preserving mappings.
\newblock {\em Algorithmica}, 2:337--361, 1987.

\bibitem{DGLRRS99}
Y.~Dodis, O.~Goldreich, E.~Lehman, S.~Raskhodnikova, D.~Ron, and
  A.~Samorodnitsky.
\newblock Improved testing algorithms for monotonicity.
\newblock In {\em RANDOM}, pages 97--108, 1999.

\bibitem{dush-miller}
B.~Dushnik and E.~Miller.
\newblock Concerning similarity transformations of linearly ordered sets.
\newblock {\em Bulletin Amer. Math. Soc.}, 46:322--326, 1940.

\bibitem{DushnikMiller}
B.~Dushnik and E.~W. Miller.
\newblock Partially ordered sets.
\newblock {\em American Journal of Mathematics}, 63:600--610, 1941.

\bibitem{h03}
W.~Hesse.
\newblock Directed graphs requiring large numbers of shortcuts.
\newblock In {\em SODA}, pages 665--669, 2003.

\bibitem{lipschitz}
M.~Jha and S.~Raskhodnikova.
\newblock Testing and reconstruction of {Lipschitz} functions with applications
  to data privacy.
\newblock Manuscript, 2010.

\bibitem{PelegSchaffer}
D.~Peleg and A.~A. Sch\"affer.
\newblock Graph spanners.
\newblock {\em Journal of Graph Theory}, 13(1):99--116, 1989.

\bibitem{Ras-survey10}
S.~Raskhodnikova.
\newblock Transitive-closure spanners: a survey.
\newblock In O.~Goldreich, editor, {\em Property Testing}, volume 6390 of {\em
  LNCS State-of-the-Art Surveys}, pages 167--196. Springer, Heidelberg, 2010.

\bibitem{SakS08}
M.~E. Saks and C.~Seshadhri.
\newblock Parallel monotonicity reconstruction.
\newblock In {\em Proceedings of the 19th Annual Symposium on Discrete
  Algorithms (SODA)}, pages 962--971, 2008.

\bibitem{SFM07}
A.~D. Santis, A.~L. Ferrara, and B.~Masucci.
\newblock Efficient provably-secure hierarchical key assignment schemes.
\newblock In {\em MFCS}, pages 371--382, 2007.

\bibitem{tho92}
M.~Thorup.
\newblock On shortcutting digraphs.
\newblock In {\em WG}, pages 205--211, 1992.

\bibitem{thorup95}
M.~Thorup.
\newblock Shortcutting planar digraphs.
\newblock {\em Combinatorics, Probability {\&} Computing}, 4:287--315, 1995.

\bibitem{thorup97}
M.~Thorup.
\newblock Parallel shortcutting of rooted trees.
\newblock {\em J. Algorithms}, 23(1):139--159, 1997.

\bibitem{yannakakis82}
M.~Yannakakis.
\newblock The complexity of the partial order dimension problem.
\newblock {\em SIAM Journal on Matrix Analysis and Applications},
  3(3):351--358, 1982.
\newblock \url{http://dx.doi.org/10.1137/0603036}.

\bibitem{Yao82}
A.~C.-C. Yao.
\newblock Space-time tradeoff for answering range queries (extended abstract).
\newblock In {\em STOC}, pages 128--136, 1982.

\end{thebibliography}
}
\appendix

%\input{appendix.tex}                % APPENDIX SECTION

%\section{Steiner Versus Non-Steiner Posets}\label{sec:steinerVersus}

%\section{Missing Proofs from Section~\ref{sec:lb-steinerspanners}: Lower Bound for General $d$ \label{app:grid-lb-generald}}

\section{Missing Proofs from Section~\ref{sec:grid-lower-bounds}: Estimating the Integral $I_d$}\label{app:grid-lb-integral}
\begin{myproof}[Proof of~\claimref{intestimate}]
To bound the integral $I_d$,  we first make a substitution $x_i = \frac{1 - t_i}{1 + t_i}$:
\begin{equation*}
I_d = \int\limits_{[-1 \dots 1]^d} \frac{dx}{\prod\limits_{1 \le i \le d} (1 + x_i) + \prod\limits_{1 \le i \le d} (1 - x_i)}.
\end{equation*}
Then we bound the denominator using the inequality $a + b \ge 2 \sqrt{ab}$ and get

\begin{equation*}
\label{intestimate}
I_d \le \int\limits_{[-1 \dots 1]^d} \frac{dx}{2 \sqrt{\prod\limits_{1 \le i \le d} (1 + x_i) \times \prod\limits_{1 \le i \le d} (1 - x_i)}} = \frac{J^d}{2},
\end{equation*}
where $J$ denotes the following integral:
\begin{equation*}
  J = \int\limits_{-1}^{1}\frac{dx}{\sqrt{1 - x^2}} = \pi.
\end{equation*}
Therefore, $I_d \le \frac{\pi^d}{2}$, as claimed.
\end{myproof}
\end{document}